\documentclass[11pt]{article}

\usepackage{amsmath,amsthm,amssymb,amscd}
\usepackage{ascmac}
\usepackage{euscript}
\usepackage{enumerate}

\usepackage[utf8]{inputenc} 
\usepackage[T1]{fontenc}    
\usepackage{hyperref}       
\usepackage{booktabs}       
\usepackage{amsfonts}       
\usepackage{bm,color,geometry,mathrsfs}
\usepackage[pdftex]{graphicx}
\def\dim{\mathop{\mathrm{dim}}\nolimits}

\def\tr{\mathop{\mathrm{tr}}\nolimits}

\def\Im{\mathop{\mathrm{Im}}\nolimits}

\def\Ker{\mathop{\mathrm{Ker}}\nolimits}

\def\id{\mathop{\mathrm{id}}\nolimits}

\def\diag{\mathop{\mathrm{diag}}\nolimits}
\def\Vol{\mathop{\mathrm{vol}}\nolimits}

\def\i{\sqrt{-1}}

\def\E{\mathbb{E}}

\def\P{\mathbb{P}}

\def\R{\mathbb{R}}

\def\Z{\mathbb{Z}}
\def\CA{\mathcal{A}}
\def\CC{\mathcal{C}}
\def\CE{\mathcal{E}}
\def\CF{\mathcal{F}}

\def\CM{\mathcal{M}}

\def\ww{w}
\def\d{\partial}
\def\lap{\mathcal{L}}

\numberwithin{equation}{section}
\newtheorem{thm}[equation]{Theorem}
\newtheorem{dfn}[equation]{Definition}

\newtheorem{prop}[equation]{Proposition}
\newtheorem{cor}[equation]{Corollary}
\newtheorem{lem}[equation]{Lemma}

\newtheorem{rem}[equation]{Remark}

\numberwithin{table}{section}
\numberwithin{figure}{section}


\title{Geometric Formulation for Discrete Points \\ and its Applications}

\author{%
  Yuuya Takayama \\
 Nikon Corporation, Japan \\
 Nikon Research Corporation of America, USA \\
 \texttt{yuuya.takayama@nikon.com} \\
}

\date{\empty}

\begin{document}
\maketitle

\begin{abstract}

We introduce a novel formulation for geometry on discrete points.
It is based on \emph{a universal differential calculus}, which gives a geometric description of a discrete set by the algebra of functions.
We expand this mathematical framework so that it is consistent with differential geometry, and works on spectral graph theory and random walks.
Consequently, our formulation comprehensively demonstrates many discrete frameworks in probability theory, physics, applied harmonic analysis, and machine learning.
Our approach would suggest the existence of an intrinsic theory and a unified picture of those discrete frameworks.

\end{abstract}

\section{Introduction}

Mathematical approaches play an essential role in understanding of practical harmonic techniques.
Though differential geometry has contributed to the theoretical studies of the Laplacian,
it does not work on discrete points, such as data.
In order to formulate its discrete analogue on a set of points, we focus on \emph{a universal differential calculus} \cite{dimakis1999discrete, grigor2015cohomology}, which has an advantage to define the exterior derivative without any additional assumption on points, likesuch as continuous models or graphs.
Since it is also possible to extend it to define the (discrete) Laplacian, this framework is naturally expected to provide a unified view among Laplacian-based algorithms in applied harmonic analysis and machine learning.
Therefore, in this paper, we aim to construct a general formulation to enable differential geometry to work on discrete points with the help of a universal differential calculus, and then study how it shows geometric relationship of frameworks in applied harmonic analysis, machine learning, and so on.

In order to build a general setting, we start from defining a differential $1$-form, a measure on functions, an inner product on $1$-forms, and the Dirichlet energy over a set of discrete points, which is regarded as a manifold.
Then, the Laplacian is immediately given as the Laplace-Beltrami operator.
It is worth emphasizing that this Laplacian is compatible with that given in spectral graph theory \cite{chung1997spectral, nica2016brief} and random walks \cite{kumagai2014random, barlow2017random}.
Finally, we define the Fourier transform and the curvature vector of an embedding, which characterize geometric aspects of points.
In summary, our formulation for differential geometry on discrete points consists of those in Table \ref{table:corresp}.

\begin{table}
  \caption{dictionary of formulation}
  \centering
  \begin{tabular}{l|l|c}
  \toprule
 manifold & discrete set $V$ & detail \\
  \midrule
 set of function & $\CA = \{f \colon V\rightarrow \R\}$ & \S \ref{subsec:universal} \\
 exterior derivative & $\d \colon \CA \rightarrow \CA\otimes \CA$ & Dfn. \ref{dfn:calculus} \\
 differential $1$-form & $\Omega_\CA^1 = \d \CA \subset \CA\otimes \CA$ & Dfn. \ref{dfn:calculus} \\
 integral of function & $\int_V \cdot d \mu\colon \CA \rightarrow \R$ & \eqref{eq:integral} \\
 inner product on functions & $\langle \cdot, \cdot \rangle_\CA\colon \CA\times\CA \rightarrow \R$ & \eqref{eq:innerproduct} \\
 inner product on $1$-forms & $\langle \cdot, \cdot \rangle_{\Omega_\CA^1}\colon \Omega_\CA^1\times \Omega_\CA^1\rightarrow \R$ & \eqref{eq:metric} \\
 Dirichlet energy & $\CE(\cdot, \cdot) \colon \CA\times \CA\rightarrow \R $ & \eqref{eq:energy} \\
 Laplacian for function & $\lap = \d^*\circ\d\colon \CA\rightarrow \CA$ & \eqref{eq:laplacian} \\
 Fourier transform & $\CF\colon \CA \rightarrow \R^{|V|}$ & \eqref{eq:fourier} \\
 embedding & $\vec{r}_d \colon V\rightarrow \R^d$ & \eqref{eq:embedding} \\
 curvature vector & $\overrightarrow{H}_{d} = - \lap \vec{r}_d \in \CA^{\oplus d}$ & \eqref{eq:curvature} \\
  \bottomrule
  \end{tabular}\label{table:corresp}
\end{table}

To show advantages of this formulation, we demonstrate three types of applications.
First, we study a graph based frameworks; spectral graph theory and random walks.
There, we review useful techniques for other applications to verify compatibility between our setting and theirs.
Second, we figure out geometric aspects of principal component analysis and classical many-body physics.
Though these frameworks are usually not explained in geometric contexts, a covariance and a force are interpreted as the Dirichlet energy and the curvature vector respectively.
Third, we understand practical applications, signal processing and manifold leaning, in applied harmonic analysis and machine learning by their relations with other frameworks.

This paper is organized as follows:
In \S \ref{sec:geom}, we explain the way to construct discrete differential geometry as in Table \ref{table:corresp}.
Since this section is discussed in an abstract manner, we summarize main concepts by matrix description in \S \ref{sec:matrix2} for the sake of the reader.
In \S \ref{sec:spectral} and \S \ref{sec:random}, we review some results from spectral graph theory and random walks, in \S \ref{sec:pca} and \S \ref{sec:physics}, we explain geometric viewpoints in principal component analysis and physics, and last we study signal processing and manifold learning in \S \ref{sec:signal} and \S \ref{sec:machinelearning} respectively.

\section{Differential geometry on discrete points}
\label{sec:geom}
In this section, we review a universal differential calculus, and then define differential geometry on a set of discrete points.

In \S \ref{subsec:universal}, we check algebraic aspects of a set of functions over discrete points.
Then, we build a geometric setting in \S \ref{subsec:metric}.
The Laplacian, the Fourier transform, and the curvature vector are introduced in \S \ref{subsec:laplace}, \S \ref{subsec:fourier} and \ref{subsec:curvature} respectively.
Their matrix description is explained in \S \ref{sec:matrix2}.

\subsection{universal differential calculus}
\label{subsec:universal}
We recall algebraic structures on functions to make sure the definition of a universal differential calculus.
See also \cite{dimakis1994discrete, grigor2015cohomology} for reference.
 
Let $V$ be a finite set.
Without loss of generality, we can assume $V = \{ 1, 2, \cdots, n\}$.
The set of functions $\{f\colon V \rightarrow \R\}$ is denoted by $\CA$, which is an $\R$-vector space in a standard manner.
It is useful to take its basis $\{e_i \in \CA\}_{i\in V}$ as $e_i(x) = \delta_{ix}$, where $\delta_{ij}$ is Kronecker's delta.
Define a product $\tilde{\sigma} \colon \CA \times\CA \rightarrow \CA$ as pointwise:
\begin{align}
\tilde{\sigma}(f, g)(x) = (f \cdot g)(x):= f(x) \cdot g(x).
\label{eq:product}
\end{align}
By bilinearity, $\tilde{\sigma}$ decomposes into two maps $\iota \colon \CA \times \CA \rightarrow \CA\otimes \CA$ and $\sigma \colon \CA \otimes \CA \rightarrow \CA$ which satisfy $\tilde{\sigma} = \sigma \circ \iota$.
Here, the tensor product $\otimes$ is over $\R$, and $\CA\otimes\CA$ is regarded as functions on $V\times V$ by $(f\otimes g)(x,y)=f(x)\cdot g(y)$.
It is easy to see these maps $\iota$ and $\sigma$ are given as
\begin{align*}
\iota(e_i, e_j) = e_i\otimes e_j, \ \sigma(e_i\otimes e_j) = \delta_{ij}e_i.
\end{align*}
We set $1_\CA$ as the constant function taking a value $1$, which is written as $1_\CA=\sum_{i\in V} e_i$.
The equation $f\cdot 1_\CA = f = 1_\CA\cdot f$ follows from the definition \eqref{eq:product}, or is checked by $e_i\cdot e_j = \delta_{ij}e_i$.
Then, we define left and right actions on $\CA\otimes\CA$ by $h\cdot(f\otimes g) := (h\cdot f)\otimes g$ and $(f\otimes g) \cdot h:= f\otimes (g\cdot h)$ respectively.
The next proposition follows:
\begin{prop}
$\CA$ is an $\R$-algebra with the product $\tilde{\sigma}$ and the unity $1_\CA$.
$\CA\otimes\CA$ is an $\CA$-bimodule.
\label{prop:bimodule}
\end{prop}
Now, we introduce a universal differential calculus.
\begin{dfn}[\cite{grigor2015cohomology}]
\label{dfn:calculus}
For $f=\sum_{i\in V} f_ie_i\in \CA$, define \emph{a differential map} $\d\colon \CA \rightarrow \CA\otimes\CA$ by
\begin{align*}
\d f:=1_\CA\otimes f - f \otimes 1_\CA =\sum_{i,j\in V}(f_j-f_i)e_i\otimes e_j,
\end{align*}
and $\Omega^1_\CA \subset \CA\otimes\CA$ as the minimal left $\CA$-submodule of $\CA\otimes\CA$ containing $\d\CA$.
The pair $(\Omega^1_\CA, \d)$ is called \emph{the universal first order differential calculus} on $\CA$. 
\end{dfn}

\begin{lem} 
$\Omega^1_\CA$ is an $\CA$-bimodule.
\end{lem}
\begin{proof}
We can see the Leibniz rule holds:
\begin{align*}
\d(f\cdot g) &= 1_\CA\otimes (f\cdot g) - (f\cdot g) \otimes 1_\CA \\
&= (1_\CA\otimes f) \cdot g - f \otimes g + f \otimes g  -  (f\cdot g) \otimes 1_\CA =\d f\cdot g + f\cdot \d g.
\end{align*}
Hence, the element $\d f\cdot g$ produced by the right action belongs to $\Omega^1_\CA$. 
\end{proof}

The differential map $\d$ can be defined on the higher tensor spaces in a similar way to the exterior derivative on manifolds \cite[\S 2]{grigor2015cohomology}.
Hence, we refer to an element of $\Omega^1_\CA$ as \emph{a $1$-form}.

\begin{lem} 
$\CA\otimes\CA$ is isomorphic to $\Omega^1_\CA\oplus \CA$ as $\CA$-bimodules.
\label{lem:isom}
\end{lem}
\begin{proof}
Notice that $e_i\cdot \d e_j = e_i\otimes e_j$ for $i\neq j$, $e_i\cdot \d e_i = -\sum_{j\in V \setminus i}e_i\otimes e_j$.
Thus, $\Omega^1_\CA$ is spanned by a basis $\{e_i\otimes e_j\mid i\neq j\}$ in $\CA\otimes \CA$.
The linear map $\sigma\colon e_i\otimes e_j \mapsto \delta_{ij}e_i$ means $\Ker \sigma=\Omega^1_\CA$,
and then we have $\CA\otimes\CA \cong \Ker \sigma \oplus\Im \sigma = \Omega^1_\CA\oplus\CA$.
\end{proof}

\subsection{measure and metric}
\label{subsec:metric}
Let $\mu$ be a measure on $V$, namely, $\mu \in \CA$ and $\mu_x:=\mu(x) > 0$ for any $x \in V$.
We define an integral on $\CA$ with respect to the measure and an inner product $\langle \cdot, \cdot \rangle_\CA\colon \CA\times\CA \rightarrow\R$ for $f=\sum f_ie_i, g=\sum g_ie_i \in \CA$:
\begin{align}
\int_V f(x)d\mu_x&:= \sum_{x\in V} \sum_{i\in V} f_i e_i(x) \mu(x) = \sum_{i\in V} f_i\mu_i, \label{eq:integral}\\
\langle f, g \rangle_\CA&:=\int_V f(x)g(x) d\mu_x = \sum_{i \in V} f_ig_i\mu_i. \label{eq:innerproduct}
\end{align}
As usual, the corresponding norm $\langle f,f\rangle_{\CA}^{1/2}$ is denoted by $\|f\|_{\CA}$, and the volume of $A\subset V$ is given by $\Vol(A) := \int_A 1_{\CA}(x) d\mu_x = \sum_{i\in A} \mu_i$.
Put a mean of $f\in \CA$ as $m_f:= \Vol(V)^{-1}\int_V f(x)d\mu_x \in \R$.
Note that the evaluation operator is represented in several ways:
\begin{align}
f_i = f(i) = e_i\cdot f = \left\langle \mu_i^{-1}e_i, f\right\rangle_{\CA}.
\label{eq:eval}
\end{align}

Let us consider an inner product on 1-forms given as a symmetric bilinear map $\langle \cdot, \cdot \rangle_{\Omega^1_\CA} \colon \Omega^1_\CA$ $\times\Omega^1_\CA \rightarrow\R$.
In this paper, we define it by
\begin{align}
\left\langle e_i\otimes e_j, e_k\otimes e_l \right\rangle_{\Omega^1_\CA} := \delta_{ik}\delta_{jl}\ww_{ij}, \label{eq:metric}
\end{align}
with $\ww_{ij} \geq 0$ which satisfies $\ww_{ij} = \ww_{ji}$ for $i, j \in V, i\neq j$.
For simplicity, we put $\ww_{ii}=0$ for $i\in V$.
This inner product satisfies the property; $\left\langle fug, v \right\rangle_{\Omega^1_\CA} = \left\langle u, fvg \right\rangle_{\Omega^1_\CA}$ for $f, g \in \CA$.
For $u = \sum_{i\neq j} u_{ij}e_i\otimes e_j, v = \sum_{i\neq j} v_{ij}e_i\otimes e_j$, we have
\begin{align*}
\langle u, v\rangle_{\Omega^1_\CA} =\sum_{i,j \in V, i\neq j} \ww_{ij}u_{ij}v_{ij}.
\end{align*}
Define a degree of the inner product as $\deg(i) := \| e_i\cdot \d e_i\|^2_{\Omega^1_\CA}$, which is often employed as a measure $\mu_i = \deg(i)$.
Since $e_i\cdot \d e_i = e_i\otimes e_i-e_i\otimes 1_{\CA}$, we get $\deg(i) = \sum_{j\in V}\ww_{ij}$.
For $f, g \in \CA$, 
\begin{align}
\CE(f,g):= \frac{1}{2}\langle \d f, \d g\rangle_{\Omega^1_\CA} = \frac{1}{2}\sum_{i,j \in V} \ww_{ij}(f_i-f_j)(g_i-g_j)
\label{eq:energy}
\end{align}
is called \emph{the Dirichlet energy} with respect to $f$ and $g$.
It is easy to check
\begin{align}
\CE(e_i,e_j) = \begin{cases}
\deg(i) & \text{for }i = j \\
-\ww_{ij} & \text{for } i\neq j \end{cases}
\label{eq:energy-weight}
\end{align}
\begin{rem}
The above inner product can be defined through a metric $\left( \cdot, \cdot \right)_{\Omega^1_\CA} \colon \Omega^1_\CA\times\Omega^1_\CA \rightarrow \CA\otimes \CA$:
\begin{align*}
\left( e_i\otimes e_j, e_k\otimes e_l \right)_{\Omega^1_\CA} := \delta_{ik}\delta_{jl}\frac{\ww_{ij}}{\mu_i \mu_j} e_i\otimes e_j,
\end{align*}
which preserves the $\CA$-bimodule structures $\left( fug, v \right)_{\Omega^1_\CA} = f \left( u, v \right)_{\Omega^1_\CA} g = \left( u, fvg \right)_{\Omega^1_\CA}$. 
By integrating it over $V\times V$, we have the inner product.
Moreover, its integration over $V$ corresponds to a dual Riemann metric $g^*\colon \Omega^1_M\times \Omega^1_M \rightarrow \CC(M)$ in differential geometry.
\end{rem}

Sometimes, it is useful to consider another basis $\{\tilde{e}_i\}:=\{e_i/\sqrt{\mu_i}\}$, which is an orthonormal basis on $(\CA, \langle \cdot, \cdot \rangle_\CA)$.
By this basis, we can represent $f\in \CA$ as
\begin{align}
f = \sum_{i\in V} \tilde{f}_i \tilde{e}_i, \text{ where } \tilde{f}_i = \langle f, \tilde{e}_i \rangle_\CA = f_i\sqrt{\mu_i}.
\label{eq:basis}
\end{align}

In general, $V$ is regarded as a set of vertices in \emph{an oriented graph} and $\ww_{ij}$ as a weight on the oriented edge $(i,j)$.
Here, we can ignore the orientation because of the condition $\ww_{ji}=\ww_{ij}$.
In this sense, $(\CA, \langle \cdot, \cdot \rangle_\CA)$ and $(\Omega^1_\CA, \langle \cdot, \cdot \rangle_{\Omega^1_\CA})$ are Hilbert spaces on the vertices and the edges respectively \cite{hein2007graph}.
When $\ww_{ij}=0$, the edge $(i,j)$ is viewed as disconnected.
If there does not exist non-empty proper subset $V' \subset V$ which satisfies $\ww_{ij}=0$ and $\ww_{ji}=0$ for all $i \in V'$ and $j \in V\setminus V'$, the graph is called connected.
\begin{rem}
The original universal differential calculus refers to a disconnected edge $(i,j)$ as a \emph{non-allowed} element $e_i\otimes e_j$, and then realizes a non complete graph as a quotient algebra of $\Omega_{\CA}^1$ by the ideal generated by non-allowed elements \emph{\cite[\S 4]{grigor2015cohomology}}.
This construction seems to describe a topology of a graph, contrary, ours focus on its metric structure.
\end{rem}
According to this convention, we often refer to $\{\ww_{ij}\}$ as (graph) weights and $(\CA, \langle \cdot, \cdot \rangle_\CA,$ $ \Omega^1_\CA, \langle \cdot, \cdot \rangle_{\Omega^1_\CA})$ as a graph.

\subsection{Laplace operator}
\label{subsec:laplace}
With the inner products given in \ref{subsec:metric}, define a co-differential $\d^*\colon \Omega^1_\CA\rightarrow\CA$ to satisfy $\langle \d^*u,f\rangle_\CA = \langle u, \d f \rangle_{\Omega^1_\CA}$ for any $u \in \Omega^1_\CA$ and $f\in \CA$.
Then, \emph{the Laplacian} $\lap\colon \CA\rightarrow \CA$ is defined by $\frac{1}{2}\d^*\circ \d$, in the same way as the Laplace-Beltrami operator in differential geometry.
Since 
\begin{align*}
\langle u, \d f\rangle_{\Omega^1_\CA} = \sum_{i,j\in V, i\neq j}\ww_{ij}u_{ij}(f_j-f_i) = \sum_{i,j \in V,i\neq j} \ww_{ij}(u_{ji} - u_{ij})f_i,
\end{align*}
we obtain
\begin{align*}
\d^*u = \sum_{i\in V} \frac{1}{\mu_i} \Bigg\{ \sum_{j \in V\setminus i} \ww_{ij}(u_{ji} - u_{ij}) \Bigg\}e_i.
\end{align*}
Thereby, the Laplacian $\lap=\frac{1}{2}\d^*\circ \d$ is represented as
\begin{align}
\lap f = \sum_{i,j\in V}\frac{\ww_{ij}}{\mu_i}(f_i - f_j)e_i  = \sum_{i\in V}\frac{\deg(i)}{\mu_i}f_ie_i - \sum_{i,j\in V}\frac{\ww_{ij}}{\mu_i}f_je_i. \label{eq:laplacian}
\end{align}
This is also known as \emph{the graph Laplacian}, as explained in \S \ref{sec:matrix2}.
By definition, we have $\CE(f,g)= \langle f, \lap g\rangle_\CA$ and the above representation follows from \eqref{eq:energy-weight} as well.
When the corresponding graph is connected, the Dirichlet energy $\CE(f,f)$ takes the minimum value $0$ if and only if $f$ is a constant function.
Since $\lap$ is self-adjoint, we can take eigenfunctions $\{v_i\}$ as follows:
\begin{align}
\lap v_i = \rho_i v_i, \ 0 = \rho_1 \leq \rho_2 \leq \cdots \leq \rho_n, \ \langle v_i,v_j \rangle_\CA = \delta_{ij}.
\label{eq:eigenfunctions}
\end{align}
Here, we see $v_1 = \Vol(V)^{-1/2}1_\CA$ and $v_i \perp 1_{\CA}$ for $i \geq 2$.

\subsection{Fourier analysis}
\label{subsec:fourier}
In the continuous setting, the Fourier transform is given by $\CF[f](\xi) = \int f(x)e^{\i\xi x}dx$, and $e^{\i\xi x}$ is an eigenfunction of the 1-dimensional Laplacian, $-\frac{d^2}{dx^2} e^{\i\xi x} = \xi^2 e^{\i\xi x}$.

On the analogy, in the graph setting, it is natural to use the eigenfunctions $\{v_i\}$ of the Laplacian $\lap$, instead of $e^{\i\xi x}$, and 
define $\CF[f]_i := \langle f, v_i \rangle_\CA \in \R$ for $f \in \CA$.
The transformation $\CF\colon \CA \ni f \mapsto \CF[f] \in \R^n$ is known as \emph{the graph Fourier transform} \cite{hammond2011wavelets}.
We call $\CF[f]_i$ the $i$-th Fourier coefficient or the $i$-th frequency.
The corresponding inverse Fourier transform is given by
\begin{align}
f = \sum_{i=1}^n \CF[f]_i v_i = \sum_{i=1}^n \langle f, v_i \rangle_\CA v_i,
\label{eq:fourier}
\end{align}
which is just the eigenfunction expansion by $\{ v_i \}$.
It is easy to see \emph{Parseval's identity} holds:
\begin{align*}
\langle f, g \rangle_{\CA} = \sum_{i=1}^n \langle f, v_i \rangle_\CA \langle g, v_i \rangle_\CA= \langle \CF[f], \CF[g] \rangle_{\R^n}.
\end{align*}
This is valid for other expansions by orthogonal functions, such as \eqref{eq:basis}.
Sometimes, the convolution operator $*g\colon \CA \rightarrow \CA$ is defined so that $\CF[f*g]_i = \CF[f]_i\CF[g]_i$ holds:
\begin{align}
f*g := \sum_{i=1}^n \CF[f]_i\CF[g]_i v_i.
\label{eq:convolution}
\end{align}
We also obtain relations $\|f\|_{\CA} = \|\CF[f]\|_{\R^n}$ and $\CF[\lap f]_i = \rho_i\CF[f]_i$.

\subsection{embedding and curvature}
\label{subsec:curvature}
Our setting so far did not use a coordinate of points in $V$, just used their indexes.
Herein, suppose that points are embedded in Euclidean space $\R^d$.
Namely, we consider a map
\begin{align}
V \ni i \mapsto (r_1(i), r_2(i), \cdots, r_d(i) ) =: \vec{r}_d(i) \in \R^d,
\label{eq:embedding}
\end{align}
where $r_s\in \CA$ and $\vec{r}_d\in \CA^{\oplus d}$.
This element $\vec{r}_d(i)$ is viewed as a coordinate for a point $i \in V$.
The Euclidean group $E(d)$ acts on $\R^d$, hence it defines a coordinate transformation $\vec{r}_d(i) \mapsto {\vec{r}\,}'_d(i) = \vec{r}_d(i)\cdot \overrightarrow{R}_d + \vec{u}_d$ for $(\overrightarrow{R}_d, \vec{u}_d) \in E(d) \cong O(d) \times \R^d$ (as a set).

In differential geometry, an embedding $\vec{r} \colon \CM \rightarrow \E^d$ induces a Riemann metric $g_{\vec{r}}$ on a manifold $\CM$, and especially determines the Laplace-Beltrami operator $\lap_{g_{\vec{r}}}$.
The normal bundle is given on $\CM$, and then the mean curvature vector $\overrightarrow{H}_{\vec{r}}$ is defined as the trace of the second fundamental form divided by $n = \dim \CM$.
Hence, the vector indicates the normal direction on each point of $\CM$, and its length is called the mean curvature.
Beltrami's formula relates those objects as
\begin{align*}
\lap_{g_{\vec{r}}} \vec{r} = -n \cdot\overrightarrow{H}_{\vec{r}}.
\end{align*}
One can refer to \cite{chen2013laplace, brakke1978motion} for mathematical details. 

Motivated by this formula,
we \emph{define} a graph curvature vector $\overrightarrow{H}_d = (H_1, H_2, \cdots, H_d)\in \CA^{\oplus d}$ of an embedding $\vec{r}_d$ by
\begin{align}
\overrightarrow{H}_{d} := -\lap \vec{r}_d, \ \ H_s := -\lap r_s, \ \ \text{ for }s =1, 2, \cdots, d.
\label{eq:curvature}
\end{align}
Unlike differential geometry, this vector does not indicate the normal direction, because it is not defined for discrete points.
Nevertheless, the vector has a special meaning in physics as explained in \S \ref{sec:physics}.
The embedding energy $\CE(\vec{r}_d, \vec{r}_d):=\sum_{s=1}^d \CE(r_s, r_s)$ is given with the curvature vector, that is,
$\CE(\vec{r}_d, \vec{r}_d) = \int_V\langle \vec{r}_d(x), -\vec{H}_d(x)\rangle_{\R^d} d\mu_x$,
and invariant by the Euclidean group action.
In some cases, it is convenient to suppose a metric $\langle \cdot, \cdot \rangle_{\Omega^1_\CA}$ is induced by an embedding as in Figure \ref{fig:objects}.
For example, we can define weights $\{\ww_{ij}\}$ by using the distance, such as
\begin{align*}
\ww_{ij} := C\cdot \exp\left( -\frac{\|\vec{r}_d(i) - \vec{r}_d(j)\|_{\R^d}^2}{2\sigma^2}\right).
\end{align*}
Several researches show this type of weights converges into a heat kernel on a manifold in the limit $|V| \rightarrow \infty$ \cite{hein2005graphs, belkin2005towards, coifman2006diffusion}.
Instead, we study another type of weights in Theorem \ref{thm:newton}.

\begin{figure}[htbp]
\begin{center}
\includegraphics[width=12cm]{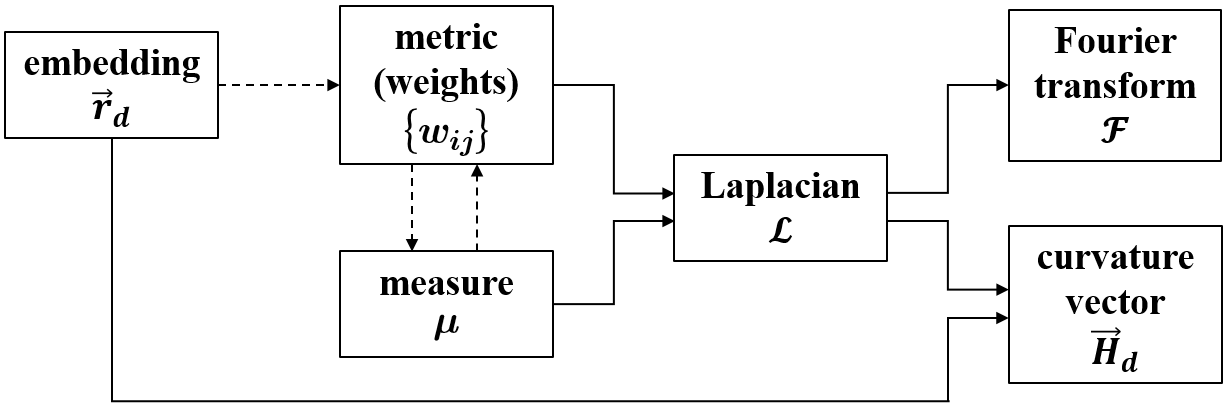}
\caption{Dependency of geometric objects. A dashed arrow means that its head object can be determined by its tail object, but not necessary.}
\label{fig:objects}
\end{center}
\end{figure}

\section{Matrix description of the geometric formulation}
\label{sec:matrix2}
In this section, we give matrix description of the formulation discussed in  \S \ref{sec:geom}. 

First we remark that our formulation contains two types of parameters in a measure $\mu$ and an inner product $\langle \cdot, \cdot \rangle_{\Omega^1_{\CA}}$ independently.
A measure $\mu$ is just given by a positive function on $V$, thus its degree of freedom is $n$.
On the other hand, an inner product $\langle \cdot, \cdot \rangle_{\Omega^1_{\CA}}$ is determined by weights $\{\ww_{ij}\}$ which satisfy
$\ww_{ij} \geq 0, \ww_{ji}=\ww_{ij}$ and $\ww_{ii}=0$, hence its degree of freedom is $n(n-1)/2$.
When they are taken on a certain relation, well-known cases appear as follows. 

Let $\bm{M} = \diag_i(\mu_i)$, $\bm{W}=\{\ww_{ij}\}_{i,j}$, and $\bm{D} = \diag_i(\deg(i))$ be $n\times n$-matrices. 

By $\{e_i\}$-basis, a function $f \in \CA$ is represented as a numerical vector $\bm{f} = {}^t\!(f_1, \cdots, f_n)$.
Then, an inner product $\langle f, g \rangle_\CA$ is written as ${}^t\!\bm{f} \bm{M} \bm{g}$, and the Dirichlet energy $\CE(f,g)=2\langle \d f, \d g\rangle_{\Omega^1_{\CA}}$ is as ${}^t\!\bm{f} (\bm{D} - \bm{W}) \bm{g}$, which does not depend on $\bm{M}$. 
The Laplacian $\lap$ is given as $\bm{M}^{-1}(\bm{D}-\bm{W})$, and its eigenvalue equation is
\begin{align*}
(\bm{D}-\bm{W})\bm{f} = \rho \bm{M}\bm{f}.
\end{align*}
The corresponding eigenvectors, denoted by an $n\times n$-matrix $\bm{V}=\{v_j(i)\}_{i,j}$, defines the Fourier transform $\bm{\CF[f]} = {}^t\!\bm{V}\bm{M}\bm{f}$ and its inverse $\bm{f} = \bm{V}\bm{\CF[f]}$, where ${}^t\!\bm{V}\bm{M}\bm{V} = \bm{I} =\bm{V} {}^t\!\bm{V}\bm{M}$.
An embedding $\vec{r}_d$ is described as an $n\times d$-matrix $\bm{R}=\{r_s(i)\}_{i,s}$, then its curvature vector $\overrightarrow{H}_d$ is as $-\bm{M}^{-1}(\bm{D}-\bm{W})\bm{R}$.

In the case of $\{\tilde{e}_i\}$-basis, since $f = \sum_{i\in V} \tilde{f}_i \tilde{e}_i$ as mentioned in \eqref{eq:basis}, we have $\langle f, g \rangle_\CA = {}^t\!\bm{\tilde{f}} \bm{\tilde{g}}$, where $\bm{\tilde{f}} = {}^t\!(\tilde{f_1}, \cdots, \tilde{f_n})$.
The Laplacian $\lap$ is described as $\bm{M}^{-1/2}(\bm{D}-\bm{W})\bm{M}^{-1/2}$, because
\begin{align*}
\lap f = \sum_{i \in V}\frac{\deg(i)}{\sqrt{\mu_i\mu_i}}\tilde{f}_i\tilde{e}_i - \sum_{i,j\in V}\frac{\ww_{ij}}{\sqrt{\mu_i\mu_j}}\tilde{f}_j\tilde{e}_i.
\end{align*}
When assume $\bm{M} = \bm{I}$ or $\bm{D}$, we obtain the known Laplacians; the combinatorial Laplacian, the random walk Laplacian, and the normalized Laplacian \cite{hein2007graph}.
However, we do not use this configuration for random walks in \S \ref{sec:random}.

The above notations are summarized in Table \ref{table:matrix2}.
\begin{table}[htbp]
  \caption{matrix description}
  \centering
  \begin{tabular}{c|c|ccc}
  \toprule
    & general by $\{e_i\}$ & $\bm{M} = \bm{I}$ & $\bm{M} = \bm{D}$ by $\{e_i\}$ & $\bm{M} = \bm{D}$ by $\{\tilde{e}_i\}$ \\
  \midrule
   $\langle f, g \rangle_\CA$ & ${}^t\!\bm{f} \bm{M} \bm{g}$ & ${}^t\!\bm{f} \bm{g}$ & ${}^t\!\bm{f} \bm{D} \bm{g}$ & ${}^t\!\bm{\tilde{f}} \bm{\tilde{g}}$ \\
   $\Vol(V)$ & $\tr\bm{M} $ & $|V| = n$ & $\tr\bm{D}$ & $\tr\bm{D}$ \\
   $\CE(f,g)$ & ${}^t\!\bm{f} (\bm{D} - \bm{W}) \bm{g}$ & ${}^t\!\bm{f} (\bm{D} - \bm{W}) \bm{g}$ & ${}^t\!\bm{f} (\bm{D} - \bm{W}) \bm{g}$ & ${}^t\!\bm{\tilde{f}} \bm{D}^{-\frac{1}{2}}(\bm{D} - \bm{W})\bm{D}^{-\frac{1}{2}} \bm{\tilde{g}}$  \\
   $\lap$ & $\bm{M}^{-1}(\bm{D} - \bm{W})$ & $\bm{D} -\bm{W}$ & $\bm{I} - \bm{D}^{-1}\bm{W}$ & $\bm{I} - \bm{D}^{-\frac{1}{2}} \bm{W} \bm{D}^{-\frac{1}{2}}$ \\
   $\CF[f]$ & ${}^t\!\bm{V}\bm{M}\bm{f}$ & ${}^t\!\bm{V}\bm{f}$ & ${}^t\!\bm{V}\bm{D}\bm{f}$ & ${}^t\!\bm{\tilde{V}}\bm{\tilde{f}}$ \\
   $\overrightarrow{H}_d$ & $\bm{M}^{-1}(\bm{W} - \bm{D})\bm{R}$ & $(\bm{W} -\bm{D})\bm{R}$ & $(\bm{D}^{-1}\bm{W} - \bm{I})\bm{R}$ & $(\bm{D}^{-\frac{1}{2}} \bm{W} \bm{D}^{-\frac{1}{2}}-\bm{I})\bm{\tilde{R}}$ \\
  \bottomrule
  \end{tabular}\label{table:matrix2}
\end{table}

\section{Application I: spectral graph theory}
\label{sec:spectral}
We review basic results about eigenvalue estimation in spectral graph theory to check its compatibility with our formulation.
These results are regarded as discrete analogues of spectral geometry and related to the graph cut problem in \S \ref{subsec:graphcut}.
See also \cite{chung1997spectral, nica2016brief} for reference.

\subsection{upper bound of eigenvalues}
First, we estimate an upper bound for the largest eigenvalue $\rho_n$.
Put $\delta := \max_{i\in V} \deg(i)/\mu_i$.
\begin{lem}
$\rho_n \leq 2\delta$.
\label{lem:upper}
\end{lem}
\begin{proof}
We have
\begin{align*}
\rho_n = \langle v_n, \lap v_n\rangle_{\CA} &= \frac{1}{2}\sum_{i,j\in V}\ww_{ij}\left(v_n(i) - v_n(j)\right)^2 \\
&\leq \sum_{i,j\in V} \ww_{ij}\left(v_n(i)^2 + v_n(j)^2\right) \leq 2\max_{i\in V} \frac{\deg(i)}{\mu_i}
\end{align*}
because $1 = \| v_n \|_{\CA}^2 = \sum_{i\in V}v_n(i)^2\mu_i$.
\end{proof}
Next, we give an upper bound for the second smallest eigenvalue $\rho_2$, which is characterized as a minimum value of $\CE(f, f) / \|f\|^2_{\CA}$ in functions $\{f \in \CA \mid f \perp 1_{\CA}\}$.
For this purpose, \emph{the isoperimetric constant} $\beta$ is useful, because it is defined in a similar way to the characterization of $\rho_2$:
\begin{align}
\beta=\min_{\emptyset \neq A\subsetneq V} \frac{\Vol(\d A)}{\Vol(A)} := \min_{\emptyset \neq A\subsetneq V}\frac{\CE(\chi_A, \chi_A)}{\|\chi_A \|_{\CA}^2},
\label{eq:isoperimetric}
\end{align}
where $\chi_A := \sum_{i\in V}e_i \in \CA$ and $A$ is taken over all subsets satisfying $\Vol(A) \leq \Vol(V)/2$.
From \eqref{eq:energy-weight}, we can check $\CE(\chi_A, \chi_A) = \sum_{i \in A, j \in A^c}\ww_{ij}$, then $\Vol(\d A^c) = \Vol(\d A)$ for a complement $A^c = V \setminus A$.
\begin{lem}
$\rho_2 \leq 2\beta$.
\end{lem}
\begin{proof}
For any $\emptyset \neq A\subsetneq V$ such that $\Vol(A) \leq \Vol(V)/2$, put $f_A = \Vol(A^c)\chi_A - \Vol(V)\chi_{A^c}$, which satisfies $f_A \perp 1_{\CA}$, then we have
\begin{align*}
\CE(f_A, f_A) & = \sum_{i\in A, j\in A^c}\ww_{ij}\left(\Vol(A) + \Vol(A^c)\right)^2= \Vol(V)^2\Vol(\d A), \\
\|f_A\|_{\CA}^2 &= \Vol(A)\Vol(A^c)^2 + \Vol(A)^2\Vol(A^c) = \Vol(V)\Vol(A)\Vol(A^c).
\end{align*}
Hence, we obtain $\rho_2 \leq  \CE(f_A, f_A) / \|f_A\|^2_{\CA} \leq 2\Vol(\d A) / \Vol(A)$ by $\Vol(V)/2 \leq \Vol(A^c)$. 
\end{proof}

\subsection{lower bound of eigenvalues}
Here, we estimate an lower bound for the second smallest eigenvalue $\rho_2$.
\begin{thm}
$\rho_2 \geq \beta^2/2\delta$.
\end{thm}
\begin{proof}
First we claim
\begin{align}
\beta\int_V f(x) d\mu_x \leq \frac{1}{2}\sum_{i,j \in V}\ww_{ij}|f_i-f_j|,
\label{eq:inequality}
\end{align}
for a positive function $f\in \CA$ such that $\Vol(\{i\in V\mid f(i)>0\}) \leq \Vol(V)/2$.
We take a sequence of subsets $\emptyset = A_0 \subsetneq A_1 \subsetneq \cdots \subsetneq A_l \subsetneq A_{l+1} = V$ so that $f$ is represented as $\sum_{s=1}^l h_s\chi_{A_s}$ by $h_s = f|_{A_s\setminus A_{s-1}} - f|_{A_{s+1}\setminus A_s} \in \R_{>0}$.
Then, we can see
\begin{align*}
\frac{1}{2}\sum_{i,j \in V}\ww_{ij}|f_i-f_j| &= \sum_{s=1}^l \sum_{i\in A_s, j\in A_s^c}\ww_{ij}h_s \\
&= \sum_{s=1}^l h_s \Vol(\d A_s) \geq \beta \sum_{s=1}^l h_s \int_V \chi_{A_s}(x) d\mu_x,
\end{align*}
as required.
Now, we can write $v_2 = g_+ - g_-$ by positive functions $g_+, g_-\in \CA$ which satisfies $\Vol(\{i\in V\mid g_+(i)>0\}) \leq \Vol(V)/2$ since $\int_V g_+(x)d\mu_x = \int_Vg_-(x)d\mu_x$.
It is easy to check $\CE(g_+, g_+) \leq \CE(v_2, g_+) = \rho_2\|g_+\|_{\CA}^2$, then applying \eqref{eq:inequality} to $f = g_+^2$, we obtain
\begin{align*}
\beta^2\|g_+\|_{\CA}^4 &\leq \frac{1}{4}\Bigl( \sum_{i,j\in V}\ww_{ij}|g_+(i) - g_+(j)|\cdot |g_+(i) + g_+(j)| \Bigr)^2 \\
&\leq \frac{1}{4}\Bigl( \sum_{i,j\in V}\ww_{ij}|g_+(i) - g_+2(j)|^2 \Bigr) \cdot \Bigl( \sum_{i,j\in V}\ww_{ij}|g_+(i) + g_+(j)|^2 \Bigr)^2 \\
&\leq \frac{1}{2}\CE(g_+, g_+) \cdot 2\sum_{i,j\in V} \ww_{ij}\left(g_+(i)^2 + g_+(j)^2\right) \leq 2\rho_2\delta \|g_+\|_{\CA}^4.
\end{align*}
In the second inequality, we used the Cauchy-Schwarz inequality: $\langle \omega_-, \omega_+ \rangle_{\Omega^1_{\CA}}^2 \leq \| \omega_- \|_{\Omega^1_{\CA}}^2\cdot\|\omega_+ \|_{\Omega^1_{\CA}}^2$ for $\omega_{\pm} = \sum_{i,j\in V, i\neq j}|v_2(i)\pm v_2(j)|e_i\otimes e_j \in \Omega^1_{\CA}$. 
\end{proof}
This theorem is called \emph{Cheeger's inequality} and its continuous analogue is known in differential geometry \cite{cheeger1969lower}.

\section{Application II: random walks}
\label{sec:random}
In this section, we deduce some notations of random walks from our formulation.
In \S \ref{subsec:heateq}, we review random walks briefly, and in \S \ref{subsec:commutedist}, we consider their connection with a geometric distance.

\subsection{heat equations}
\label{subsec:heateq}
For $c > 0$, putting $S_c := \id - c^{-1} \lap \colon \CA\rightarrow \CA$, we have
\begin{align}
S_cf = \sum_{i\in V}\frac{c\mu_i - \deg(i)}{c\mu_i}f_ie_i + \sum_{i,j\in V}\frac{\ww_{ij}}{c\mu_i}f_je_i = \sum_{i,j\in V}\frac{\theta_{ij}}{c\mu_i}f_je_i,
\label{eq:integralop}
\end{align}
for $f\in \CA$, where we put $\theta_{ii} = c\mu_i - \deg(i)$ and $\theta_{ij}=\ww_{ij}$ for $i, j \in V$.
Needless to say, the eigenvalue decomposition of $S_c$ is given as
\begin{align*}
S_c v_i= \lambda_i v_i, \ \lambda_i = 1 - \frac{\rho_i}{c}, \ 1 = \lambda_1 \geq \lambda_2 \geq \cdots \geq \lambda_n \geq 1-2\frac{\delta}{c},
\end{align*}
by \eqref{eq:eigenfunctions} and Lemma \ref{lem:upper}.
Besides, $S_c f$ is described as
\begin{align}
S_c f = \sum_{i=1}^n\lambda_i \langle f, v_i\rangle_\CA v_i.
\label{eq:mercer}
\end{align}
\begin{prop}
For $k \in \Z_{\geq 0}$ and $t \in \R_{\geq 0}$, define operators
\begin{align*}
P_k := S_c^k, \ \ Q_t := \sum_{k=0}^\infty \frac{e^{-t}t^k}{k!} P_k.
\end{align*}
In addition, put $p_k: =P_kf_0$ and $q_t=Q_tf_0$ for any $f_0\in \CA$.
Then, $p_k$ and $q_t$ satisfy the discrete and continuous time heat equations
\begin{align}
p_{k+1} - p_k &= -\frac{1}{c}\lap p_k, \ \ p_0 = f_0 \label{eq:discretehe}\\
\frac{d}{dt}q_t &= -\frac{1}{c}\lap q_t, \ \ q_0 = f_0 \label{eq:continuoushe}
\end{align}
respectively.
\label{prop:heateq}
\end{prop}
\begin{proof}
The discrete time heat equation follows from $p_{k+1} - p_k = (S_c-\id)p_k$.
Besides, for $q_t$, we have,
\begin{align*}
\frac{d}{dt}q_t = \sum_{k=0}^\infty \frac{e^{-t}t^k}{k!} \left(-p_k + p_{k+1}\right) = \left(S_c-\id\right)\sum_{k=0}^\infty \frac{e^{-t}t^k}{k!} p_k.
\end{align*}
Hence, $q_t$ is a solution of the continuous time heat equation.

The latter part is also checked by \eqref{eq:mercer}.
We have $p_k = \sum_{i=1}^n \lambda_i^k \langle f_0,v_i\rangle_\CA v_i$, thus,
\begin{align*}
q_t = \sum_{i=1}^n \sum_{k=0}^{\infty} e^{-t}\frac{t^k\lambda_i^k}{k!} \langle f_0,v_i\rangle_\CA v_i = \sum_{i=1}^n e^{ -t \rho_i/c} \langle f_0,v_i\rangle_\CA v_i.
\end{align*}
This gives $Q_t = \exp(-t\lap/c)$, then the assertion immediately follows.
See also \S \ref{subsec:filtering}.
\end{proof}
In this sense, $S_c$ is viewed as \emph{an integral operator of a heat kernel}, which is given as $\sum_{i=1}^n\lambda_iv_i\otimes v_i \in \CA\otimes \CA$.
If we impose the condition $c \geq \delta$, then we have $(S_ce_y)(x) = \theta_{xy} \geq 0 $ for any $x, y \in V$.
Thereby, a function $y \mapsto (S_ce_y)(x)$ defines a discrete probability distribution because of $S_c1_{\CA} = 1_{\CA}$.
In random walk settings,
\begin{align*}
P(x, y):= (S_ce_y)(x) = \frac{\theta_{xy}}{c\mu_i}, \ \ p_k(x,y):=\frac{(P_ke_y)(x)}{\mu(y)}
\end{align*}
is known as \emph{the transition probability and transition density} respectively \cite{kumagai2014random, barlow2017random}.
Here, $(P_ke_y)(x) = \P^x(X_k = y)$ means the probability of transitioning from $x$ to $y$ in $k$ steps.

\subsection{commute time distance}
\label{subsec:commutedist}
Let $\tau_+ := \min \{ k \ge 1 \mid X_k = y \}$ be the first hitting time and $m(x,y) := \E^x [ \tau_+]$ be its expectation.
It is easy to see
\begin{align*}
m(x,y) &= P(x,y) + \sum_{z \in V\setminus y}P(x,z)(1+m(z,y)) \\
&=1 + \sum_{z \in V}P(x,z)m(z,y) - P(x,y)m(y,y).
\end{align*}
Notice that we can rearrange the above equation as 
\begin{align}
m(\cdot, y) = 1_{\CA} + S_cm(\cdot, y)-m(y,y)S_ce_y,
\label{eq:firsthit}
\end{align}
where $m(\cdot, y) \in \CA$.
\begin{thm}[\cite{fouss2007random}]
Put $T(\cdot,y) := m(\cdot,y) - m(y,y)e_y \in \CA$.
Then, we have
\begin{align*}
n(x,y):=\frac{T(x,y) + T(y,x)}{c\Vol(V)} = \sum_{i=2}^n\frac{1}{\rho_i}(v_i(y)-v_i(x))^2,
\end{align*}
for any $x, y\in V$.
\end{thm}
\begin{proof}
From the equation \eqref{eq:firsthit}, we have 
\begin{align*}
c^{-1}\lap T(\cdot, y) = (\id - S_c)(m(\cdot,y) - m(y,y)e_y) = 1_{\CA} - m(y,y)e_y.
\end{align*}
By taking the inner product with $v_i$ and using \eqref{eq:eval}, we get
\begin{align*}
0 = c^{-1}\rho_1\langle v_1, T(\cdot, y)\rangle_{\CA} &= \Vol(V)^{1/2} - m(y,y)\mu_y v_1(y), &&\text{for }i =1, \\ 
c^{-1}\rho_i\langle v_i, T(\cdot, y)\rangle_{\CA} &= - m(y,y)\mu_y v_i(y), &&\text{otherwise.}
\end{align*}
The first equation means $m(y,y)\mu_y = \Vol(V)$, and the second equation leads to the Fourier coefficients, hence we obtain
\begin{align*}
T(\cdot, y) &= \sum_{i = 1}^n \CF[T(\cdot, y)]_iv_i \\
&= a_01_{\CA} - \sum_{i=2}^n\frac{c\Vol(V)}{\rho_i}v_i(y)v_i = \sum_{i=2}^n\frac{c\Vol(V)}{\rho_i}v_i(y)(v_i(y)1_{\CA} - v_i).
\end{align*}
At the last equality, we used the fact $T(y,y)=0$ to determine a constant $a_0$.
\end{proof}
$n(x,y)$ is known as \emph{the commute time distance} divided by $c\Vol(V)$, and equal to the Euclidean distance via the embedding 
\begin{align}
V \ni x\mapsto (\rho_2^{-1/2}v_2(x), \rho_3^{-1/2}v_3(x), \cdots). \label{eq:walk}
\end{align}
This is regarded as one of branches of \emph{the Laplacian eigenmaps} \cite{belkin2002laplacian},
\begin{align}
V \ni x\mapsto (v_2(x), v_3(x), \cdots),
\label{eq:lapmaps}
\end{align}
and it links with PCA is studied by \cite{fouss2007random}.
We also see these relations in the following section.

\section{Application III: PCA}
\label{sec:pca}
Here, we study a geometric aspect of the empirical covariance, and then consider a principal component analysis as a branch of manifold learning.

\subsection{random variable and embedding}
\label{subsec:random}
Let $(\Omega, \mathcal{F}, P)$ be a probability space, and we consider a random variable $X = (X_1, X_2, \cdots, $ $X_d) \colon \Omega \rightarrow \R^d$.
A mean $\E[X_s]\in \R$ and a covariance $C(X_s, X_t)\in \R$ are defined by
\begin{align*}
\E[X_s] := \int_\Omega X_s(\omega) dP_\omega, \ \ C(X_s, X_t) := \E[X_sX_t] - \E[X_s]\E[X_t].
\end{align*}
Now, we regard $(V, 2^V, \mu/\Vol(V))$ as a probability space, then it follows that $X_s \in \CA$ and $m_{X_s} = \E[X_s]$.
Moreover, we can characterize a covariance as well.
\begin{thm}
Take $\ww_{ij} = \mu_i\mu_j/\Vol(V)^2$ for $i\neq j$.
Then a covariance coincides with the Dirichlet energy: $C(X_s,X_t) = \CE(X_s,X_t)$.
\label{thm:bridge}
\end{thm}
\begin{proof}
We have
\begin{align*}
\E[X_sX_t] - \E[X_s]\E[X_t] &= \frac{1}{2}\sum_{i, j\in V}\frac{\mu_i\mu_j}{\Vol(V)^2}\left(X_s(i) - X_s(j)\right)\left(X_t(i) - X_t(j)\right) \\
& = \CE(X_s,X_t),
\end{align*}
as required.
\end{proof}
\begin{cor}
In the above setting, we have $\lap X_s = (X_s - \E[X_s]\cdot 1_{\CA})/\Vol(V)$.
In particular, $\rho_i = \Vol(V)^{-1}$ for $i\geq 2$ in \eqref{eq:eigenfunctions}.
\label{cor:randomlap}
\end{cor}
\begin{proof}
From \eqref{eq:laplacian}, we have
\begin{align*}
\lap f = \sum_{i,j\in V}\frac{\mu_{j}}{\Vol(V)^2}(f_i - f_j)e_i  = \frac{1}{\Vol(V)}(f - \E[f]\cdot 1_{\CA}). 
\end{align*}
For $i \geq 2$, we get $\E[v_i] = 0$, then $\lap v_i = \Vol(V)^{-1}v_i$.
\end{proof}
Therefore, if $X$ is a centered variable, then the corresponding curvature vector is $0$, and besides if variables $\{X_s\}$ are independent, then they forms orthogonal eigenfunctions of the Laplacian. 

\subsection{principal component}
\label{subsec:pca}
For a random variable $X\colon\Omega \rightarrow \R^d$, put $\overline{X}(\omega) := X(\omega) - \E[X]$.
Take eigenfunctions of the covariance matrix $C_{\overline{X}}$, that is, $C_{\overline{X}}u_s = \alpha_su_s, u_s \in \R^d$ and $\alpha_1 \geq \alpha_2 \geq \cdots \geq \alpha_d$.
This is viewed as a diagonalization by the Euclidean group $E(d)$ acting on $\R^d$.
In this setting, the eigenfunction expansion of $\overline{X}(\omega)$ is given as
\begin{align}
\overline{X}(\omega) = \sum_{s=1}^d \xi_s (\omega) u_s, \text{  where } \ \xi_s(\omega) = \left\langle \overline{X}(\omega), u_s \right\rangle_{\R^d} \in \R.
\label{eq:pca}
\end{align}
The coefficient $\xi_s$ is a map $\Omega \rightarrow \R$, hence, a random variable.
A straightforward calculation shows
\begin{align}
\E[\xi_s] = 0, \ \ \E[\xi_s \xi_t] = \alpha_s\delta_{st},
\label{eq:pcaprop}
\end{align}
for $1 \leq s, t \leq d$.
This means the coefficients $\{\xi_s\}$ are not correlated each other, and their covariances decrease as the index becomes larger.
Therefore, a mapping $\Omega \ni \omega \mapsto (\xi_1(\omega), \xi_2(\omega), \cdots)$ is called \emph{a principal component analysis} (PCA), where only the first few principal terms are usually taken.

Let us give geometric interpretations for PCA by Theorem \ref{thm:bridge}.
Since the random variables $\{\xi_s\}$ are centered and independent as \eqref{eq:pcaprop}, $\{\xi_s\}$ consist of as eigenfunctions of the Laplacian given in Corollary \ref{cor:randomlap}.
Hence, PCA is regarded as a special case of the Laplacian eigenmaps \eqref{eq:lapmaps} or the embedding given in \eqref{eq:walk}.

In addition, PCA is reformulated as a problem to maximize the left hand side of
\begin{align}
\left\langle C_{\overline{X}}a, a \right\rangle_{\R^d} = \sum_{s,t=1}^d a_sa_t \CE(X_s,X_t),
\label{eq:lpp}
\end{align}
for $a = (a_1, a_2, \cdots, a_d) \in \R^d$ under the condition $\| a \|_{\R^d} = 1$.
Then the first principal term is written as  $\xi_1 = \sum_{s=1}^d a_sX_s\in \CA$.
By contrast, the right hand side means the embedding energy of one dimensional subspace in $\R^d$ given by $r = \sum_{s=1}^d a_sX_s\in \CA$.
This shows a relation between PCA and \emph{Locallity Preserving Projections} (LPP), which provides another embedding by \emph{minimizing} \eqref{eq:lpp} under the condition $\| r \|_{\CA} = 1$ \cite{he2004locality}.

\section{Application IV: many-body physics}
\label{sec:physics}
In this section, we regard embedded points in $\R^d$ as point mass in a classical many-body system.
We show the force coincides with the curvature vector by taking special weights.

\subsection{Hooke's law of spring}
Let $\{\vec{r}_d(i) \mid i \in V \}$ be $n$ points in $\R^d$ having masses $\{m_i > 0\}$.
We assume all pair of points $\left(\vec{r}_d(i), \vec{r}_d(j)\right)$ is connected by a zero-length spring with force constant $k_{ij} \geq 0$.
For simplicity, put $k_{ii}=0$ for $i\in V$.
In this case, Hooke's law states the force $\overrightarrow{F} \in \Gamma(V;\R^d)$ and the potential $U\in \R$ are respectively given as
\begin{align*}
\overrightarrow{F}(i) = \sum_{j\in V} k_{ij}(\vec{r}_d(j) - \vec{r}_d(i)), \ \ U = \frac{1}{2}\sum_{i,j \in V} k_{ij} \|\vec{r}_d(i) - \vec{r}_d(j)\|^2_{\R^d}.
\end{align*}
\begin{thm}
Take $\ww_{ij} = k_{ij}$ and $\mu_i = m_i$ for $i,j \in V$.
Then we have
\begin{align*}
\overrightarrow{F}(i) = m_i\overrightarrow{H}_d(i), \ \ U = \CE(\vec{r}_d, \vec{r}_d).
\end{align*}
\end{thm}
\begin{proof}
The assertion immediately follows from \eqref{eq:laplacian}, \eqref{eq:curvature}, and \eqref{eq:energy}.
\end{proof}
Its equations of motion is represented as
\begin{align*}
m_i\frac{d^2}{dt^2}\vec{r}_d(i) = \overrightarrow{F}(i) = -m_i\lap \vec{r}_d(i).
\end{align*}
Since the eigenfunctions of $\lap$ does not depend of $\vec{r}_d$, we can see the $i$-th Fourier coefficient has frequency $\rho_i$:
\begin{align*}
\frac{d^2}{dt^2}\CF[\vec{r}_d]_i = \frac{d^2}{dt^2}\langle \vec{r}_d, v_i\rangle_{\CA} = -\rho_i\langle \vec{r}_d, v_i\rangle_{\CA} = -\rho_i\CF[\vec{r}_d]_i.
\end{align*}
If we fix some points $A \subset V$, the stable positions of other free points $A^c$ are calculated by solving the Dirichlet problem
\begin{align}
\begin{cases}
\lap \vec{r} (i) = 0, \text{ for } i \in A^c, \\
\vec{r}|_A = \vec{r}_d|_A.
\end{cases}
\label{eq:dirichlet}
\end{align}
This idea is used to obtain a smooth surface in point cloud processing \cite{vollmer1999improved}.
The Dirichlet problem for a function is considered in random walks and machine learning, and applied to image processing \cite{grady2006random} and semi-supervised learning in \cite{zhu2003semi} respectively.

\subsection{Newton's law of gravitation}
\label{subsec:newton}
Now, we assume $d \geq 3$.
Let $\{\vec{r}_d(i) \mid i \in V \}$ be $n$ points in $\R^d$ as before.
The gravitational potential $\Phi \in \CC(\R^d)$ is considered to follow the Poisson equation in $\R^d$ with a boundary condition at infinity,
and its general solution gives
\begin{align*}
\Phi_i(\vec{r}) = -C\sum_{j \in V \setminus i} \frac{Gm_j}{\| \vec{r} - \vec{r}_d(j)\|^{d-2}_{\R^d}}.
\end{align*}
Here, $\Phi_i$ means the gravitational potential around $\vec{r} = \vec{r}_d(i)$ caused by $\{\vec{r}_d(j) \mid j \in V \setminus i \}$.
The corresponding gravitational field $\vec{g} \in \Gamma(V; \R^d)$ are given by $\vec{g}(i) = - \overrightarrow{\nabla}_{\R^d} \Phi_i|_{\vec{r}=\vec{r}_d(i)}$, hence we have
\begin{align}
\vec{g}(i) = -C(d-2)\sum_{j\in V\setminus i} \frac{Gm_j(\vec{r}_d(i) - \vec{r}_d(j))}{\| \vec{r}_d(i) - \vec{r}_d(j) \|^{d}_{\R^d}}.
\label{eq:field}
\end{align}
The gravitational potential energy $U \in \R$ is described as
\begin{align*}
U = \frac{1}{2}\sum_{i\in V} m_i\Phi_i(\vec{r}_d(i)) = -\frac{C}{2}\sum_{i, j \in V, i\neq j} \frac{Gm_im_j}{\| \vec{r}_d(i) - \vec{r}_d(j)\|^{d-2}_{\R^d}}.
\end{align*}
These physical concepts defined on $\R^d$ are directly described as those on $V$.

\begin{thm}
Take $\ww_{ij} = CGm_im_j/\| \vec{r}_d(i) - \vec{r}_d(j)\|^{d}_{\R^d}$ for $i, j \in V, i \neq j$ otherwise $0$, and $\mu_i = m_i$.
We have
\begin{align*}
\vec{g}(i) = (d-2)\overrightarrow{H}_d(i), \ \ U = -\CE(\vec{r}_d, \vec{r}_d).
\end{align*}
\label{thm:newton}
\end{thm}
\begin{proof}
Under the assumption, for $i \in V$ we have
\begin{align*}
-(d-2)\lap \vec{r}_d(i) &= -(d-2)\sum_{j\in V\setminus i} \frac{\ww_{ij}}{m_i}(\vec{r}_d(i) - \vec{r}_d(j)) = \vec{g}(i), \\
\CE(\vec{r}_d, \vec{r}_d) &= \frac{1}{2}\sum_{i,j\in V, i\neq j}\ww_{ij}\|\vec{r}_d(i) - \vec{r}_d(j)\|^2_{\R^d} = -U,
\end{align*}
as required.
\end{proof}
Remark that the above weights depend on positions $\vec{r}_d$ unlike the case of Hooke's law.
Hence, to compute the stable positions, we cannot use direct calculation like \eqref{eq:dirichlet}.
In this situation, we can use an iterative method like \eqref{eq:flow} with variable weights, which is viewed as \emph{the mean shift algorithm} \cite{comaniciu2002mean} in machine learning.

In these case, weights are given as physical constants defining the system, and the curvature vector \eqref{eq:curvature} indicates the force (up to constant) in a totally discrete manner, like general relativity in a continuous setting. 

\begin{rem}
From a similar viewpoint, the Coulomb potential is studied as weights of a kernel by \cite{hochreiter2003coulomb}, and their exponent is different from that in Theorem \ref{thm:newton}.
\end{rem}

\section{Application V: signal processing}
\label{sec:signal}
In this section, we review several harmonic techniques used in signal processing.
Here, the Fourier transform and the curvature vector play important roles.

\subsection{filtering}
\label{subsec:filtering}
The convolution operator given in \eqref{eq:convolution} is viewed to weight the Fourier coefficients of $f$ by those of $g$.
This can be generalized as a filtering on a frequency domain.
In other words, the eigenvalues $\{ \rho_i \}$ are understood as the frequencies, thus, with a continuous function $g\colon \R_{\geq 0} \rightarrow \R$, we define a filtering operator $T_g:= g(\lap)\colon \CA \rightarrow \CA$ by
\begin{align}
T_g f := \sum_{i=1}^n \CF[f]_i g(\lap)v_i = \sum_{i=1}^n g(\rho_i)\CF[f]_i v_i.
\label{eq:filter}
\end{align}
The operator $Q_t = \exp(-t\lap/c)$ given in Proposition \ref{prop:heateq} is also understood in this sense.
In particular, a low-pass filter is realized by taking a decreasing function $g$ such that $g(0) = 1$ and $\lim_{x\rightarrow \infty}g(x) = 0$.
Practically, it is convenient to approximate $g$ with the Chebyshev polynomials to avoid calculating the eigenfunctions \cite{shuman2011chebyshev, hammond2011wavelets}.
The filtering operator for an embedding $\vec{r}_d$ is proposed as \emph{the manifold harmonic transform} \cite{vallet2008spectral} for analyzing point clouds, which is described as $\vec{r}_3 \in \CA^{\oplus 3}$.
In order to obtain useful functions on a graph, the filtering operators \eqref{eq:filter} is modified to add learnable parameters in the graph deep learning \cite{bruna2013spectral, henaff2015deep}.
In another context, the filtering operator is extended to define the graph wavelet transform \cite{hammond2011wavelets}.

\subsection{smoothing}
\label{subsec:smoothing}
In order to construct useful filtering operators, let us perturb an embedding $\vec{r}_d$ to decrease the embedding energy $\CE(\vec{r}_d,\vec{r}_d)$.
\begin{align}
{\vec{r}\,}_d^{(k+1)} =  {\vec{r}\,}_d^{(k)} - \epsilon \lap{\vec{r}\,}_d^{(k)} = {\vec{r}\,}_d^{(k)} + \epsilon {\overrightarrow{H}\,}_d^{(k)}.
\label{eq:flow}
\end{align}
This is seen as the discrete time heat equation \eqref{eq:discretehe}, or the explicit Euler-scheme for the continuous time heat equation \eqref{eq:continuoushe}.
From a geometric viewpoint, it corresponds to the mean curvature flow \cite{brakke1978motion}.
Here, we remark that weights $\{\ww_{ij}\}$ are usually fixed during the iterations even if they depend on the initial ${\vec{r}\,}_d^{(0)} = \vec{r}_d$.

Since iterations of \eqref{eq:flow} asymptotically lead to a constant embedding, which means just one point, it is important to prevent from shrinking.
Some improving methods are proposed, for example, using the second ordered Laplacian combining with a growing up process \cite{taubin1995signal}, considering the implicit Euler-scheme alternatively \cite{hein2007manifold}, or both \cite{desbrun1999implicit};
\begin{align*}
{\vec{r}\,}_d^{(k+1)} =
\begin{cases}
\left( \id - \epsilon' \lap \right)\left( \id - \epsilon \lap \right){\vec{r}\,}_d^{(k)}, \\
{\vec{r}\,}_d^{(k)} - \epsilon \lap {\vec{r}\,}_d^{(k+1)}, \\
{\vec{r}\,}_d^{(k)} - \epsilon \lap^2 {\vec{r}\,}_d^{(k+1)},
\end{cases}
\end{align*}
where $\epsilon>0$ and $\epsilon + \epsilon' < 0$.
They are described respectively as $g(t) = (1-\epsilon' t)(1-\epsilon t), (1- \epsilon t)^{-1}$ and $(1- \epsilon t^2)^{-1}$ in terms of \eqref{eq:filter}, and then behave like low-path filters.
Other geometric flows are also well studied in \cite{xu2006discrete}.

Nowadays, image filtering is also understood in this context:
\begin{align*}
f' = f - \frac{1}{c}\lap f = S_{c}f
\end{align*}
as in \eqref{eq:integralop}.
The weights $\{ \theta_{ij} \}$ are usually decided by pixel's location and values \cite{milanfar2012tour}.

\section{Application VI: manifold learning}
\label{sec:machinelearning}
In machine learning, dimension reduction techniques based on graphs are sometimes called \emph{manifold leaning}.
In this section, we review some of them from viewpoints of the graph cut problem and the energy minimization.

\subsection{graph cut problem}
\label{subsec:graphcut}
The graph cut problem aims to find clusters in graph by minimizing several cutting loss functions \cite{von2007tutorial}.
They have a similar form to the isoperimetric constant \eqref{eq:isoperimetric} and are closely related with the eigenvalue problem as discussed in \S \ref{sec:spectral}.

In this paper, we set the problem as minimization of the following function
\begin{align*}
\text{Loss}_{\text{GC}}\left(\{\chi_{A_l}\}_{l=1}^{k}\right) := \sum_{l=1}^{k} \frac{\CE(\chi_{A_l}, \chi_{A_l})}{\| \chi_{A_l} \|^2_\CA},
\end{align*}
where $\chi_{A_l}=\sum_{i\in A_l} e_i$ and $\{A_l\}_{l=1}^{k}$ is a $k$-partition of $V$, which satisfies $\sqcup_{l=1}^{k} A_l = V$ (disjoint union) and $A_l \neq \emptyset$ for any $l$.
By \eqref{eq:mercer}, we have
\begin{align*}
\langle \chi_{A_l}, S_{c}\chi_{A_l}\rangle_{\CA} = \sum_{i=1}^n \lambda_i \langle v_i, \chi_{A_l} \rangle_{\CA}^2 = \sum_{i=1}^n \lambda_i \iint_{A_l\times A_l} v_i(x)v_i(y) d\mu_xd\mu_y,
\end{align*}
hence $\lap = c\id - cS_c$ and $\|v_i\|_{\CA} = 1$ lead to
\begin{align*}
\sum_{l=1}^{k} \frac{\CE (\chi_{A_l}, \chi_{A_l})}{\| \chi_{A_l}\|^2_\CA} = \sum_{i=1}^n\rho_i + \sum_{l=1}^k \sum_{i=1}^n  \frac{c\lambda_i}{2\Vol(A_l)}\iint_{A_l\times A_l} (v_i(x) - v_i(y))^2  d\mu_xd\mu_y.
\end{align*}
Furthermore, in general, we get
\begin{align*}
\frac{1}{2}\iint_{A\times A} \left( f(x) - f(y)\right)^2 d\mu_x d\mu_y & = \frac{1}{2}\iint_{A\times A} \left( (f(x) - m_{f,A}) - (f(y) - m_{f,A})\right)^2 d\mu_x d\mu_y \\
& \ \ + \iint_{A\times A} \left( f(x) - m_{f,A} \right) \left( f(y) - m_{f,A} \right) d\mu_x d\mu_y \\
& = \Vol(A) \int_{A} \left( f(x) - m_{f,A} \right)^2 d\mu_x,
\end{align*}
for any $f \in \CA$ and $m_{f,A} := \Vol(A)^{-1}\int_A f(x) d\mu_x \in \R$. 
Therefore, by taking large $c$, we can obtain the minimum partition by the $k$-means algorithm for
\begin{align*}
\left\{\left(\sqrt{\lambda_1}v_1(x), \sqrt{\lambda_2}v_2(x), \cdots, \sqrt{\lambda_n}v_{n}(x)\right)\in \R^n \mid x\in V\right\},
\end{align*}
which is nothing but \emph{the kernel $k$-means algorithm} \cite{scholkopf1998nonlinear}.
This fact is first shown by \cite{dhillon2004kernel}.
On the other hand, the graph cut problem is often translated into the $k$-means on the Laplacian eigenmaps \eqref{eq:lapmaps}, which is known as \emph{spectral clustering} \cite{von2007tutorial}.

In both cases, mapping methods play an essential role, and such methods are researched as dimension reduction.
We already explained its examples; the commute time distance embedding \eqref{eq:walk}, PCA \eqref{eq:pca}, and LPP \eqref{eq:lpp}.
LLE we discuss in the following subsection is also one of examples.

\subsection{weight learning}
\label{subsec:weightlearning}

The algorithms in the previous subsection are highly dependent on choice of weights $\{\ww_{ij}\}$, which determines eigenvalues and eigenfunctions.
In order to avoid its trial-and-error process, several researches propose methods to learn weights.
Their basic idea is to minimize an energy function with respect to weights under some assumptions.

In the Locally Linear Embedding (LLE) \cite{roweis2000nonlinear}, weights are determined to minimize the energy function $\sum_{i\in V}\|\vec{r}_d(i) - \sum_{j\in V}\tilde{\ww}_{ij} \vec{r}_d(j)\|_{\R^d}^2$ under the condition $\sum_{j\in V} \tilde{\ww}_{ij}  = 1$.
For these weights $\{\tilde{\ww}_{ij}\}$, we can take $\mu_i = \deg(i) = 1$ for all $i\in V$.
Hence, the energy function is equivalent to the length of the curvature vector for the embedding:
\begin{align*}
\sum_{i=1}^n \Bigg\| \vec{r}_d(i) - \sum_{j\in V}\tilde{\ww}_{ij}\vec{r}_d(j) \Bigg\|_{\R^d}^2 = \sum_{s=1}^{d} \sum_{i=1}^n \left( -(\lap r_s)_i \right)^2 = \sum_{s=1}^{d} \| H_s \|^2_\CA. 
\end{align*}
Besides, the LLE requires minimizing the embedding cost function $\sum_{i\in V}\|\vec{y}_{d'}(i) - \sum_{j\in V}\tilde{\ww}_{ij}$ $\vec{y}_{d'}(j)\|_{\R^{d'}}^2=\sum_{l=1}^{d'} \| \lap y_l \|_{\CA}^2$ for $d' < d$ under the condition $\langle y_l, y_{l'} \rangle_{\CA} = \delta_{ll'}$, which gives the Laplacian eigenmaps \eqref{eq:lapmaps} again.

\section{Conclusion}
We introduced a formulation based on a universal differential calculus and differential geometry, and explained several frameworks to analyze discrete points.
These demonstrations would show our formulation has a potential to understand various discrete frameworks and develop new harmonic techniques by combining graph theory, probability theory, spectral geometry, and topological techniques \cite{wasserman2018topological}.

\section*{Acknowledgement}
The author would like to thank Satoshi Takahashi, Tetsuya Koike, Yosuke Otsubo, Chikara Nakamura for useful discussion and constant encouragement.
He is also grateful to Bausan Yuan, Ping-Wei Chang, Shruthi Kubatur, Henry Chau and Pranav Gundewar for their advices.

\small

\bibliographystyle{plain}


\end{document}